\documentclass[11pt,a4paper]{article}
\usepackage{times}
\usepackage{soul}
\usepackage{url}
\usepackage[hidelinks]{hyperref}
\usepackage[small]{caption}
\usepackage{graphicx}
\usepackage{amsmath}
\usepackage{amsthm}
\usepackage{amssymb}
\usepackage{booktabs}
\usepackage{algorithm}
\usepackage{algorithmic}
\usepackage{fullpage}

\usepackage[round]{natbib}

\newtheorem{theorem}{Theorem}

\newtheorem{lemma}[theorem]{Lemma}

\newcommand{\ud}{\,\text{d}}
\newcommand{\Rs}{\mathcal{R}^s_F}
\newcommand{\Rh}{\mathcal{R}^h_F}
\newcommand{\Aux}{\mathcal{A}_F}
\newcommand{\Q}{\mathcal{Q}_F}
\newcommand{\V}{\mathcal{V}_F}
\newcommand{\R}{\mathcal{R}^*_F}
\newcommand{\Rone}{\mathcal{R}^1_F}
\newcommand{\Rtwo}{\mathcal{R}^2_F}
\newcommand{\E}{\mathbb{E}}
\newcommand{\one}{\mathbb{I}}
\newcommand{\vv}{\mathbf{v}}
\newcommand{\x}{\mathbf{x}}
\newcommand{\p}{\mathbf{p}}
\newcommand{\bb}{\mathbf{b}}
\newcommand{\hatv}{\widehat{v}}
\newcommand{\eps}{\varepsilon}
\newcommand{\shortcite}[1]{\citeyearpar{#1}}
\newcommand{\moc}{M_{oc}}
\newcommand{\mpi}{M_{pi}}

\renewcommand{\cite}{\citep}

\title{Selling an item among a strategic bidder and a profiled agent}
\author{
\textbf{Ioannis Caragiannis}\\
\small{Aarhus University, Denmark}\\
\small{iannis@cs.au.dk}
\newline
\and
\textbf{Georgios Kalantzis}\\
\small{University of Edinburgh, United Kingdom}\\
\small{G.Kalantzis@sms.ed.ac.uk}
}
\date{}

\allowdisplaybreaks
\begin{document}

\maketitle

\begin{abstract}
We consider the fundamental scenario where a single item is to be sold to one of two agents. Both agents draw their valuation for the item from the same probability distribution. However, only one of them submits a bid to the mechanism. The other agent is {\em profiled}, i.e., the mechanism receives a {\em prediction} for her valuation, which can be true or false. Our goal is to design mechanisms for selling the item that make as much revenue as possible in cases of a correct or incorrect prediction. As a benchmark for proving our revenue-approximation guarantees, we use the maximum expected revenue that can be obtained by a strategic and an honest bidder. We study two mechanisms. The first one yields optimal revenue when the prediction is guaranteed to be correct and a constant revenue approximation when the prediction is incorrect, assuming that the agent valuations are drawn from a monotone hazard rate (MHR) distribution. The second mechanism ignores the prediction for the second agent and simulates the revenue-optimal mechanism when no bid information for the bidders is available. We prove, again assuming that valuations are drawn from MHR distributions, that this mechanism achieves a constant revenue approximation guarantee compared to our revenue benchmark. The MHR assumption is necessary; we show that there are non-MHR but regular probability distributions for which no constant approximation of our revenue benchmark is possible.
\end{abstract}

\section{Introduction}\label{sec:intro}

Traditionally, selling an item is approached in one of two ways. The first is to use an auction: potential buyers submit their bids to an auctioneer, who decides who should get the item and at which price. Typically, the price depends on the bid of all potential buyers. This approach leverages competition to maximize revenue. The second is to use a fixed price, often determined using statistical information about the potential buyers. Personalized information may differentiate the price (e.g., via discounts to specific buyers). Here, achieving high revenue due to the simplicity of pricing is the goal.

In this paper, we study a {\em hybrid} setting that combines elements of the two approaches. We consider two potential buyers of a single available item: one agent who submits a bid and another agent who responds to a fixed price. The first acts strategically, aiming to buy the item at the lowest price possible and maximizing her utility. The second agent has no expressive power herself; a profiling service provides information about her preference. How should we design high-revenue mechanisms for selling an item in such an environment?

More specifically, our model assumes that both potential buyers draw their valuation for the item independently from a common probability distribution. The first is a {\em strategic bidder} and mechanisms need to ensure her truthful behavior (i.e., so that her bid is equal to her true valuation for the item). The second is a {\em profiled agent}. This means that the mechanism is just given a {\em prediction} of her valuation for the item. This prediction can be correct, i.e., equal to the true valuation of the profiled agent, or incorrect. Our goal is to design mechanisms that make as high a revenue as possible in this Bayesian setting, regardless of whether the prediction is correct or not.

What does high revenue mean in our setting? One approach could be to assume that both agents are strategic and use classical auction theory to compute the optimal revenue as the maximum virtual social welfare~\cite{myerson1981optimal}. However, this estimate can be very conservative. As mentioned above, the profiled agent has no expressive power and, hence, cannot affect the outcome of the mechanism. Furthermore, in the case the prediction of her valuation is correct, it is like having a strategic and an honest bidder. The maximum expected revenue that can be extracted by a strategic and an honest bidder will be our revenue benchmark in this work. We would like to design mechanisms which achieve as high a fraction of the revenue benchmark as possible, regardless of whether the prediction is correct or not. This is an ambitious goal that goes beyond traditional auction design.

Borrowing our terminology from the recent trend of {\em learning-augmented} algorithms (or algorithms with predictions; see, e.g., the survey by~\citealt{mitzenmacher2022algorithms}), we will use the terms {\em consistency} and {\em robustness} to evaluate our mechanisms. Both terms quantify how far from the revenue benchmark the revenue extracted by a mechanism is. Consistency is the ratio of the revenue benchmark over the revenue extracted by the mechanism when given a correct prediction of the profiled agent valuation. Robustness is the ratio of the revenue benchmark over the worst-possible revenue that is extracted by the mechanism which utilizes an incorrect prediction. 

We specifically study two mechanisms. The first one is the {\em optimal-consistency mechanism} $\moc$. Here, we can assume that the mechanism receives as input a possibly strategic bid by the strategic bidder and an honest bid by the profiled one. Myerson's theory tells us that the maximum expected revenue we can extract from the strategic bidder is her expected virtual valuation. However, from the honest bidder, we can extract the entire valuation as revenue simply by asking the agent to pay her honest bid. Hence, the mechanism $\moc$ is defined as follows. If the virtual bid of the strategic bidder is higher than the prediction of the valuation for the profiled agent, the strategic bidder gets the item paying the inverse virtual value of the prediction. Otherwise, the profiled agent gets the item, paying her predicted (true) value. Notice that in case of a true prediction, the item is always sold. However, if the true valuation of the profiled agent is smaller than the prediction, buying the item would result in negative utility. To bound the robustness of mechanism $\moc$, we pessimistically assume that no revenue is ever extracted from the profiled agent when the prediction is incorrect.

Our second mechanism $\mpi$ is designed with this pessimistic scenario in mind and completely ignores the prediction. Instead, to decide the price offered to the profiled agent, the information about the probability distribution of valuations (and predictions) is used. Hence, the mechanism $\mpi$ is designed using a backward induction argument. If the item is to be sold to the profiled agent, the price will be set so that its virtual value is $0$; by definition, this maximizes the expected revenue of the profiled agent if she is offered the item. Then, to maximize expected revenue, the price offered to the first agent should be such that she gets the item when her virtual valuation is higher than the price offered to the profiled agent. Since this mechanism ignores the prediction, it has the same consistency and robustness.

Even in the very simple setting we consider, the revenue analysis is highly non-trivial. Even though we get closed-form expressions for the revenue benchmark, computing the worst-case robustness bound among a large class of distributions is very challenging. Specifically for mechanism $\mpi$, where a close-form expression of its expected revenue is available as well, we cannot compare its expected revenue directly to the revenue benchmark. Instead, since we know that this is the revenue-maximizing mechanism that ignores the prediction, we use other prediction-ignoring mechanisms as proxies to bound the robustness of mechanism $\mpi$.

In our technical results, we assume that the valuations are drawn from monotone hazard rate (MHR) probability distributions. This assumption is necessary since otherwise the robustness of any mechanism can be unbounded. For the optimal-consistency mechanism, we show a constant robustness bound of at most $10.873$ for any MHR distribution. Our best robustness bound is $2.575$. It turns out that these bounds are worst-case values of functions that depend on a parameter of the distribution. We present trade-offs showing that the robustness of the optimal-consistency mechanism is considerably lower than $10.873$ for distributions of best robustness close to $2.575$. Our analysis exploits the properties of MHR probability distributions, including a technical lemma of~\citet{dhangwatnotai2010revenue}, which relates the expected valuation and the revenue extracted from such distributions.

The rest of the paper is structured as follows. We present a survey of the related literature in the rest of this section. Details on Myerson's theory that will be useful in our analysis are then given in Section~\ref{sec:prelim}. Important quantities for our analysis as well as a key technical lemma follows in Section~\ref{sec:quantities}. The analysis of the optimal-consistency mechanism is given in Section~\ref{sec:opt-consistency} and the robustness of our prediction-ignoring mechanism follows in Section~\ref{sec:ignore-prediction}. The robustness trade-off between the two mechanisms is presented in Section~\ref{sec:tradeoff}. We conclude with open problems in Section~\ref{sec:open}.

\paragraph{Related work.}
Learning-augmented algorithms have been in the epicenter of the recent algorithmic research literature, with numerous related contributions. Many classical problems have been reconsidered, and new algorithms, enhanced with (possibly erroneous) machine-learned predictions about their input, are designed and analyzed with respect to the achievable consistency and robustness. Representative problem domains include data structures, online and approximation algorithms for combinatorial optimization, streaming and sublinear algorithms, and many more. For a more detailed exposition see the early survey by Mitzenmacher and
Vassilvitskii~\shortcite{mitzenmacher2022algorithms} and the online repository \url{https://algorithms-with-predictions.github.io/}.

In algorithmic game theory and computational social choice, the concept of prediction has been considered for problems related to mechanism design \cite{agrawal2022learning,BGT23,balkanski2023online,colini2024trust,christodoulou2024mechanism,balkanski2024randomized,QNS24,amanatidis2025online}, the price of anarchy of cost sharing \cite{gkatzelis2022improved,christodoulou2025improving}, and the distortion of voting \cite{berger2023optimal,filos2025utilitarian}. More closely to our work, the prediction framework has been considered in auction environments. In particular, Xu and Lu \shortcite{XL22} were the first to study revenue-maximizing auctions with the challenging benchmark of the highest agent valuation, while Caragiannis and Kalantzis \shortcite{caragiannis2024randomized} developed randomized auctions in the same setting. On the other hand, Lu et al. \shortcite{lu2024competitive} studied competitive auctions and digital goods with predictions and Gkatzelis et al. \shortcite{gkatzelis2025clock} studied clock auctions with unreliable advice.

Our decision of using a revenue benchmark in order to evaluate the efficiency of our mechanisms is similar in spirit with the framework of competitive auctions, initiated with the work of Goldberg et al. \shortcite{goldberg2006competitive} and adopted by Lu et al. \shortcite{lu2024competitive}. Finally, close to our work is the study of revenue-maximizing auctions with a single sample by Dhangwatnotai et al. \shortcite{dhangwatnotai2010revenue}, where they initiated the study of \textit{random reserve} prices. More generally, the field of Bayesian mechanism design, which uses extensively statistical information about the agent valuations, is surveyed by Hartline \shortcite{hartline2013mechanism}.

\section{Preliminaries}\label{sec:prelim}
We begin by presenting the auction basics that will be useful in our study. Further details can be found in classical textbooks and auction theory literature, e.g. see \cite{roughgarden2016twenty,hartline2013mechanism} and \cite[Chapter 5]{krishna2009auction}. 

In the most standard setting of single-item auctions, there is a set of $n$ potential buyers (or {\em agents}) and a single item for sale. Each agent has a private valuation $v_i$ for the item. An {\em auction mechanism} takes as input a bid from each agent (or {\em bidder}) and decides who should get the item and at which price. In particular, the mechanism is identified by a pair $(\x,\p)$ of an allocation and a payment function. Both functions have $n$ components. For $i\in [n]$, we denote by $x_i(\bb)$ the indicator variable denoting whether the item is given to agen $i$ when the bids submitted by the agents form the bid vector $\bb$. A natural constraint for the single-item setting is that $\sum_{i\in [n]}{x_i(\bb)}\leq 1$, i.e., at most one potential buyer can get the item.

An important desideratum in the auction literature is truthfulness. In his seminal paper, \citet{myerson1981optimal} proved that a mechanism is truthful if and only if the allocation function is monotone. Then, under mild normalization assumptions, there is a unique payment rule $\p$ which, together with the allocation rule $\x$ define a truthful mechanism. In the single-item setting, the potential buyer who gets the item is the only one who pays the mechanism and her payment is equal to the minimum bid that would still give the agent the item, assuming that the other agents keep their bids.

An important goal in auction design is {\em revenue maximization}. The classical {\em Bayesian} setting is the most suitable to achieve this goal. Here, each agent $i$ draws her valuation for the item from a known probability distribution $F_i$, independently on the other agents. It turns out that the expected revenue that a truthful mechanism extracts from agent $i$ (i.e., the expected payment that the agent submits to the auctioneer) is equal to the expected {\em virtual valuation} of the agent when she gets the item. Formally, for every agent $i$ and every vector $\vv_{-i}$ of valuations by the other agents, 
\begin{align*}
\E_{v_i\sim F_i}[p_i(\vv)] &= \E_{v_i\sim F_i}[\phi_i(v_i)\cdot x_i(\vv)].
\end{align*}
The virtual valuation $\phi_i(\cdot)$ for agent $i$ is defined as $\phi_i(z)=z-\frac{1-F_i(z)}{f_i(z)}$, where $F_i$ and $f_i$ are the cumulative distribution function and probability density function of the probability distribution from which agent $i$ draws her valuation.

Then, to design an auction that maximizes the expected revenue, we equivalently have to optimize the expected virtual welfare
\begin{align*}
\E\left[\sum_{i\in [n]}{\phi_i(v_i)\cdot x_i(\vv)}\right],
\end{align*}
where $\vv=(v_1, ..., v_n)$ and the expectation is taken over $v_i\sim F_i$ for $i\in [n]$. Assuming regularity of the probability distributions, which implies that $\phi_i(z)$ is non-decreasing in $z$, the revenue-maximizing auction is defined as follows. If all virtual bids are negative, no agent gets the item. Otherwise, the item is given to the agent with the highest virtual bid at a price whose virtual value is equal to either $0$ or to the second highest virtual bid, whichever is higher. This is nothing more than maximizing the sum above, under the constraint that at most one of the $x_i(
v)$ indicator variables can be equal to $1$.

Our setting has only two agents $s$ and $h$. Both draw their valuation for the item independently from a common probability distribution $F$. Define the {\em hazard rate} of $F$ as $h(z)=\frac{f(z)}{1-F(z)}$. We will be assuming that $h(z)$ is non-decreasing in $z$ and will refer to $F$ as a {\em monotone hazard rate} (MHR) probability distribution. For example, the uniform and the exponential distributions are MHR. Any MHR distribution is {\em regular}, having non-decreasing virtual value.

Mechanisms for allocating the item interact differently with the two agents. The agent $s$ submits a bid $b_s$ to the mechanism. For agent $h$, we are given a prediction $\hatv$ of her valuation. The prediction follows the distribution $F$, can be correlated to the true valuation $v_h$ of agent $h$, but is independent on the valuation $v_s$ of agent $s$. The bid $b_s$ and the prediction $\hatv$ form the input to the mechanism. A mechanism is defined by the allocation function $\x=(x_s,x_h)$ so that $x_s(b_s,\hatv),x_h(b_s,\hatv)\in \{0,1\}$ and $x_s(b_s,\hatv)+x_h(b_s,\hatv)\leq 1$.

The arguments discussed above for standard single-item auctions hold for the strategic bidder $s$ as well. To ensure her truthful behavior, it is necessary and sufficient that the allocation function $x_s(z,\hatv)$ is non-decreasing with $z$. Then, the payment $p_s(b_s,\hatv)$ agent $s$ should submit to the mechanism when she gets the item (i.e., $x_s(b_s,\hatv)=1$) is equal to the minimum value $b$ so that $x_s(b,\hatv)=1$. Then, the expected revenue that is extracted by agent $s$ is her expected virtual valuation when she gets the item.

The prediction $\hatv$ is not provided by agent $h$ and, hence, there is no truthfulness issue for her. The mechanism can decide to offer the item to agent $h$ at some price and agent $h$ will decide to buy and pay the price or ignore the offer, depending on whether the price is below or above her true valuation $v_h$. The ideal scenario is when the prediction is correct, i.e., $\hatv$ is always equal to the value $v_h$ agent $h$ draws from $F$. In this case, agent $h$ would always accept a price equal to $\hatv$.

We are ready to design the revenue-optimal mechanism for our setting when the prediction is true. We will refer to this mechanism as the {\em optimal-consistency mechanism} $\moc$. This is essentially the revenue-optimal auction for a strategic and an honest bidder. All we have to do is to maximize the quantity
\begin{align*}
    \E_{v_s,v_h\sim F}\left[x_s(b_s,\hatv)\cdot \phi(b_s)+x_h(b_s,\hatv)\cdot \hatv\right]
\end{align*}
Hence, mechanism $\moc$ works as follows. If the virtual bid $\phi(b_s)$ of agent $s$ is no less than the prediction, agent $s$ gets the item at the price $\phi^{-1}(\hatv)$ (this is the smallest value with virtual value not smaller than $\hatv$). Otherwise, bidder $h$ gets the item at the price $\hatv$.

In case the prediction is incorrect, agent $h$ may reject the offer of the item at the price $\hatv$. Then, the whole second term in the expression of the expected revenue above can be lost. We will bound the {\em robustness} of the optimal-consistency mechanism $\moc$ by bounding the ratio of the expectation above over the expectation of the first among the two terms only.

So, mechanism $\moc$ achieves the maximum expected revenue that can be extracted by a strategic and an honest bidder if the prediction is correct but considerably lower revenue if the prediction is incorrect. It seems reasonable then to ignore the prediction entirely and, since there is no other information available, to rely on the statistical information available for the profiled agent. Then, if agent $h$ is offered the item, this should be done at the price $\phi^{-1}(0)$, which maximizes the expression $z(1-F(z))$ denoting the expected revenue extracted from the agent who draws their valuation from distribution $F$ and is offered the item at a price of $z$. Let us define $\R=\phi^{-1}(0)\cdot (1-F(\phi^{-1}(0)))$ to be the optimal expected revenue that can be extracted by agent $h$ with a fixed price. Then, we can verify that agent $s$ should be offered the item at the price of $\phi^{-1}(\R)$ for a total expected revenue of
\begin{align}\label{eq:rev-of-mpi}
    \phi^{-1}(\R)\cdot (1-F(\phi^{-1}(\R)))+F(\phi^{-1}(\R))\cdot \R.
\end{align}
We will refer to this revenue-optimal mechanism among the mechanisms that ignore the prediction as mechanism $\mpi$. To bound its robustness, we need to bound the ratio of the maximum expected revenue for a strategic and an honest bidder over the expression in (\ref{eq:rev-of-mpi}). Even though both quantities have closed-form expressions, comparing them will be possible only indirectly.

\section{Analysis warm-up}\label{sec:quantities}
In our analysis for mechanisms $\moc$ and $\mpi$, we use several quantities associated with a probability distribution $F$. The first one is simply the expected value $\V$ drawn from $F$:
\begin{align*}
    \V&=\int_0^{\infty}{z\cdot f(z)\ud{z}}.
\end{align*}
The next quantity is the expected revenue $\Q$ extracted from a single agent who draws her valuation from $F$ when she is offered the item at an independently selected random price also selected from $F$:
\begin{align*}    
    \Q&=\int_0^{\infty}{z\cdot (1-F(z))\cdot f(z)\ud{z}}.
\end{align*}
In our analysis, we will also use the auxiliary quantity $\Aux$ defined as 
\begin{align*}
        \Aux&=\int_0^{\infty}{\left(\phi^{-1}(z)-z\right)\cdot\left(1-F(\phi^{-1}(z))\right)\cdot f(z)\ud{z}}.
\end{align*}
We also define $\chi_F=\Aux/\V$.

We will use the following statement proved by \citet{dhangwatnotai2010revenue}, which relates the quantities defined above; recall that $\R$ is the optimal expected revenue that can be extracted by
a single agent.
\begin{lemma}[\citealt{dhangwatnotai2010revenue}]\label{lem:r-q-vs-v}
    Let $F$ be a monotone hazard rate probability distribution returning non-negative values. Then $\Q \geq \V/4$ and $\R\geq \V/e$.
\end{lemma}

We proceed by proving the following technical lemma. It generalizes the well-known property that an MHR distributions returns a value above Myerson's reserve price with probability at least $1/e$, i.e., $1-F(\phi^{-1}(0))\geq 1/e$. 
\begin{lemma}\label{lem:quantiles}
Let $F$ be a monotone hazard rate probability distribution returning non-negative values, and $\phi$ the corresponding virtual valuation function. Then, for every $z\geq 0$, 
\begin{align*}
    F(\phi^{-1}(z)) &\leq 1 - \frac{1}{e}\cdot (1-F(z)).
\end{align*}
\end{lemma}

\begin{proof}
Let $h$ be the hazard rate of $F$ and $H$ the cumulative hazard rate defined as $H(z)=\int_0^z{h(y)\ud{y}}$. Then, it can be verified that $1-F(z)=e^{-H(z)}$ for every $z\geq 0$. Hence,
\begin{align*}
    \frac{1-F(\phi^{-1}(z))}{1-F(z)} &= \frac{e^{-H(\phi^{-1}(z))}}{e^{-H(z)}} = \exp\left(\int_0^z{h(y)\ud{y}}-\int_0^{\phi^{-1}(z)}{h(y)\ud{y}}\right)\\
    &=\exp\left(-\int_z^{\phi^{-1}(z)}{h(y)\ud{y}}\right).
\end{align*}
The proof of the lemma will follow by upper-bounding the integral at the RHS of the above equation by $1$. Indeed, by the MHR assumption, we have that $h(y)$ is non-decreasing in $y$ and, hence, 
\begin{align*}
    \int_z^{\phi^{-1}(z)}{h(y)\ud{y}} &\leq h(\phi^{-1}(z))\cdot (\phi^{-1}(z)-z)=1.
\end{align*}
The last equality follows since, by the definition of the virtual valuation, $z=\phi(\phi^{-1}(z))=\phi^{-1}(z)-\frac{1}{h(\phi^{-1}(z))}$.
\end{proof}

\section{The optimal-consistency mechanism}\label{sec:opt-consistency}
We denote by $\Rs$ and $\Rh$ the expected revenue that mechanism $\moc$ extracts from the strategic and the honest bidder, respectively. Recall that mechanism $\moc$ sells to the strategic bidder when $\phi(v_s)\geq v_h$ and to the honest bidder when $\phi(v_s)<v_h$. In the first case, the mechanism charges the critical bid $\phi^{-1}(v_h)$ to the winning strategic bidder. In the second case, the mechanism charges the winning honest bidder her valuation. Hence, the expected revenue $\Rs$ extracted from the strategic bidder is
\begin{align}\nonumber
    \Rs &= \E_{v_h\sim F}\E_{v_s\sim F}[
    \phi^{-1}(v_h) \one \{\phi(v_s)\geq v_h\}]\\\nonumber
    &= \int_0^{\infty}{\left(\int_{\phi^{-1}(v_h)}^{\infty}{\phi^{-1}(v_h)f(v_s)\ud{v_s}}\right)\cdot f(v_h)\ud{v_h}}\\\label{eq:rev-strategic}
    &= \int_0^{\infty}{\phi^{-1}(z)\cdot(1-F(\phi^{-1}(z)))\cdot f(z)\ud{z}}
\end{align}
and the expected revenue $\Rh$ extracted from the honest bidder is
\begin{align}\nonumber
    \Rh &= \E_{v_h\sim F}\E_{v_s\sim F}[v_h\,\one\{v_h\geq \phi(v_s)\}]\\\nonumber
    &= \int_0^{\infty}{\left(\int_0^{\phi^{-1}(v_h)}{v_h\cdot  f(v_s)\ud{v_s}}\right)\cdot f(v_h)\ud{v_h}}\\\label{eq:rev-honest}
    &=\int_0^{\infty}{z\cdot F(\phi^{-1}(z))\cdot f(z)\ud{z}}.
\end{align}

\begin{lemma}\label{lem:optimal-rev}
    The expected revenue of mechanism $\moc$ when applied to a strategic and an honest bidder who draw their valuations independently according to the probability distribution $F$ is $\V\cdot (\chi_F+1)$.
\end{lemma}

\begin{proof}
    By the discussion above, the expected revenue of mechanism $\moc$ is $\Rs+\Rh$. Using inequalities (\ref{eq:rev-strategic}) and (\ref{eq:rev-honest}), and the definitions of quantities $\Aux$, $\V$, and $\chi_F$, we get
    \begin{align*}
        \Rs+\Rh &= \int_0^{\infty}{\phi^{-1}(z)\cdot(1-F(\phi^{-1}(z)))\cdot f(z)\ud{z}} +\int_0^{\infty}{z\cdot F(\phi^{-1}(z))\cdot f(z)\ud{z}}\\
        &=\int_0^{\infty}{\left(\phi^{-1}(z)-z\right)\cdot(1-F(\phi^{-1}(z)))\cdot f(z)\ud{z}} +\int_0^{\infty}{z\cdot f(z)\ud{z}}\\
        &=\Aux+\V=\V\cdot (\chi_F+1),
    \end{align*}
as desired.
\end{proof}

Next, consider the execution of mechanism $\moc$ that uses a prediction $\hatv$ for the valuation of the profiled agent instead of her true valuation $v_h$ submitted as an honest bid. If the prediction is correct, i.e., $\hatv=v_h$, the profiled agent can be thought of as an honest bidder and the prediction $\hatv$ as her honest bid $v_h$. If $\phi(v_s)\leq \hatv$, the mechanism will offer the item to the profiled bidder at the price of $\hatv$. If the prediction is correct, i.e., $\hatv=v_h$, the profiled bidder will accept to buy as her utility is non-negative. Obviously, the mechanism extracts optimal revenue in this case and has optimal consistency. If the prediction is incorrect, the profiled bidder will accept to buy only if her true valuation is not lower than the prediction, i.e., if $v_h\geq \hatv$. We use these observations together with the Equations (\ref{eq:rev-strategic}) and (\ref{eq:rev-honest}) above to prove the following statement.

\begin{lemma}\label{lem:rev-with-wrong-prediction}
    The expected revenue of mechanism $\moc$ when applied to a stratefic bidder and a profiled agent who draw their valuations from the monotone hazard rate probability distribution $F$ is at least $\V\cdot \left(\chi_F+\frac{1}{4e}\right)$.
\end{lemma}

\begin{proof}
We will pessimistically assume that no revenue at all is extracted from the profiled bidder when the prediction is incorrect.\footnote{Notice that the revenue from the profiled bidder can indeed be very close to $0$. For example, imagine valuations drawn from the uniform distribution in $[0,1]$ and a correlated prediction defined as $\hatv=v_h+\eps$ when $v_h\in (0,1-\eps]$, $\hatv=v_h+\eps-1$ when $v_h\in (1-\eps,1]$, and $\hatv=0$ when $v_h=0$. The probability that the true valuation $v_h$ of the profiled agent is not lower than the prediction and, equivalently, the probability that she will accept to buy the item at a price of $\hatv$, is only $\eps$, which can become negligibly small.} Since the prediction $\hatv$ is independent on the valuation $v_s$ and (like $v_h$) follows the probability distribution $F$, and the item is allocated to the strategic bidder when $\phi(v_s)\geq \hatv$, the expected revenue the mechanism extracts from the strategic bidder is equal to the quantity $\Rs$ defined in Equation (\ref{eq:rev-strategic}). Using Equation (\ref{eq:rev-strategic}), the definitions of quantities $\Aux$, $\Q$, and $\chi_F$, the fact $\phi^{-1}(z)\geq z$, and the first inequality from Lemma~\ref{lem:r-q-vs-v}, we get
\begin{align*}
    \Rs &= \int_0^{\infty}{\phi^{-1}(z)\cdot(1-F(\phi^{-1}(z)))\cdot f(z)\ud{z}}\\
    &=\int_0^{\infty}{\left(\phi^{-1}(z)-z\right)\cdot(1-F(\phi^{-1}(z)))\cdot f(z)\ud{z}}+\int_0^{\infty}{z\cdot (1-F(\phi^{-1}(z))\cdot f(z)\ud{z}}\\
    &\geq \Aux +\frac{1}{e}\cdot \int_0^\infty{z\cdot (1-F(z))\cdot f(z)\ud{z}}\\
    &= \Aux+\frac{1}{e}\cdot \Q\geq \Aux+\frac{1}{4e}\cdot \V=\V\cdot \left(\chi_F+\frac{1}{4e}\right),
\end{align*}
as desired.
\end{proof}

\begin{theorem}\label{thm:robustness-of-m-star}
    The robustness of mechanism $\moc$ when applied on a strategic bidder and a profiled agent (with a possibly incorrect prediction) who draw their valuations from the monotone hazard rate probability distribution $F$ is at most $4e\approx 10.873$.
\end{theorem}

\begin{proof}
    We use the definition of robustness for mechanism $\moc$ as its expected revenue obtained when applied to a strategic and an honest bidder, divided by its expected revenue obtained when applied to a strategic bidder and a profiled agent with an incorrect prediction. By Lemmas~\ref{lem:optimal-rev} and~\ref{lem:rev-with-wrong-prediction}, we get that this ratio is at most $\frac{\chi_F+1}{\chi_F+\frac{1}{4e}}\leq 4e$. The inequality follows since $\chi_F\geq 0$.   
\end{proof}

\section{Achieving better robustness}\label{sec:ignore-prediction}
We now present an analysis of mechanism $\mpi$. Recall that mechanism $\mpi$ is the mechanism that extracts the maximum expected revenue among all mechanisms that ignore the prediction of the profiled agent and just use the information they have for the distribution $F$. We will bound the expected revenue of mechanism $\mpi$ by the expected revenue of the mechanism $M_1$ which also ignores the prediction of the profiled agent and is defined as follows. Mechanism $M_1$ selects a random sample $z$ from $F$ and offers the item to the strategic agent at a price of $\phi^{-1}(z)$. If the strategic agent does not buy the item, the mechanism offers the item to the profiled agent at the price $\phi^{-1}(0)$. 

For the sample $z$ drawn by the mechanism, the strategic agent buys the item with probability $1-F(\phi^{-1}(z))$, i.e., the probability her valuation is above the offered price of $\phi^{-1}(z)$. If the strategic agent does not buy the item (this happens with probability $F(\phi^{-1}(z))$), an optimal expected revenue of $\R$ is extracted from the profiled agent.

Denoting by $\Rone$ the expected revenue of mechanism $M_1$, we have
\begin{align}\label{eq:rone}
    \Rone &= \int_0^{\infty}{\left(\phi^{-1}(z)\cdot \left(1-F(\phi^{-1}(z))\right)+F(\phi^{-1}(z) \right)\cdot \R \cdot f(z) \ud{z}}.
\end{align}
The next two lemmas provide lower bounds on the expected revenue of mechanism $M_1$ and, consequently, $\mpi$.

\begin{lemma}\label{lem:m1-rev-1}
    The expected revenue of mechanism $M_1$ when applied to a stratefic bidder and a profiled agent who draw their valuations from the monotone hazard rate probability distribution $F$ is at least $\V\cdot \left(\chi_F+\frac{3}{4e}\right)$.
\end{lemma}

\begin{proof}
Working with the definition of $\Rone$ in Equation~(\ref{eq:rone}), we have
\begin{align*}
    \Rone &= \int_0^{\infty}{\left(\phi^{-1}(z)\cdot \left(1-F(\phi^{-1}(z))\right)\right.}\\
    &\quad\quad\quad\quad{\left.+F(\phi^{-1}(z) \right) \cdot \R\cdot f(z) \ud{z}}\\
    &=\int_0^{\infty}{\left(\phi^{-1}(z)-z\right)\cdot \left(1-F(\phi^{-1}(z))\right)\cdot f(z) \ud{z}}+\int_0^{\infty}{z\cdot \left(1-F(\phi^{-1}(z))\right)\cdot f(z)\ud{z}}\\
    &\quad\,+\R\cdot \int_0^{\infty}{F(\phi^{-1}(z))\cdot f(z)\ud{z}}\\
    &\geq \Aux + \frac{1}{e}\cdot \int_0^{\infty}{z\cdot (1-F(z))\cdot f(z)\ud{z}}+\R\cdot \int_0^{\infty}{F(z)\cdot f(z)\ud{z}}\\
    &=\Aux+\frac{1}{e}\cdot \Q+\frac{1}{2}\cdot \R\geq \Aux+\frac{3}{4e}\cdot \V=\V\cdot \left(\chi_F+\frac{3}{4e}\right).
\end{align*}
as desired. The first inequality follows by the definition of $\Aux$, Lemma~\ref{lem:quantiles}, and the fact $\phi^{-1}(z)\geq z$. The two last equalities follow by the definitions of quantities $\Q$, $\R$, and $\chi_F$. The second inequality follows by Lemma~\ref{lem:r-q-vs-v}. \end{proof}

\begin{lemma}\label{lem:m1-rev-2}
    The expected revenue of mechanism $M_1$ when applied to a stratefic bidder and a profiled agent who draw their valuations from the monotone hazard rate probability distribution $F$ is at least $\V\cdot \left(\frac{1}{e}\cdot \chi_F+\frac{5}{4e}-\frac{1}{2e^2}\right)$.
\end{lemma}

\begin{proof}
Notice that, due to the regularity of distribution $F$,  $\phi^{-1}(z)$ is non-decreasing in $z$. Then, it holds that
$\phi^{-1}(z)\geq \phi^{-1}(0)\geq \R$ and, using Lemma~\ref{lem:quantiles}, we get
\begin{align}\nonumber
    &\phi^{-1}(z)\cdot \left(1-F(\phi^{-1}(z))\right)+F(\phi^{-1}(z))\cdot \R\\\nonumber
    &=\phi^{-1}(z)-F(\phi^{-1}(z))\cdot (\phi^{-1}(z)-\R)\\\nonumber
    &\geq \phi^{-1}(z)-\left(1-\frac{1}{e}\cdot (1-F(z))\right)\cdot (\phi^{-1}(z)-\R)\\\nonumber
    &= \frac{1}{e}\cdot \left(\phi^{-1}(z)-z\right)\cdot (1-F(z)) +\frac{1}{e}\cdot z\cdot (1-F(z))+ \R\cdot \left(1-\frac{1}{e}+\frac{1}{e}\cdot F(z)\right)\\\label{eq:per-term-bound}
    &\geq \frac{1}{e}\cdot \left(\phi^{-1}(z)-z\right)\cdot (1-F(\phi^{-1}(z)))+\frac{1}{e}\cdot z\cdot (1-F(z))+ \R\cdot \left(1-\frac{1}{e}+\frac{1}{e}\cdot F(z)\right).
\end{align}
The last inequality is true since $\phi^{-1}(z)\geq z$. 
Working again with the definition of $\Rone$ from Equation (\ref{eq:rone}) and using inequality (\ref{eq:per-term-bound}), the definitions of quantities $\Aux$, $\Q$, and $\chi_F$, and Lemma~\ref{lem:r-q-vs-v}, we have
\begin{align*}
    \Rone &= \int_0^{\infty}{\left(\left(\phi^{-1}(z)\cdot \left(1-F(\phi^{-1}(z))\right)+F(\phi^{-1}(z)\right)\cdot \R\right)\cdot f(z) \ud{z}}\\
    &\geq \frac{1}{e}\cdot \int_0^{\infty}{\left(\phi^{-1}(z)-z\right)\cdot (1-F(\phi^{-1}(z)))\cdot f(z)\ud{z}}\\ &\quad\,+\frac{1}{e}\cdot \int_0^{\infty}{z\cdot (1-F(z))\cdot f(z)\ud{z}}+ \R\cdot \left(1-\frac{1}{e}+\frac{1}{e}\cdot \int_0^{\infty}{F(z)\cdot f(z)\ud{z}}\right)\\
    &= \frac{1}{e}\cdot \Aux +\frac{1}{e}\cdot \Q + \left(1-\frac{1}{2e}\right)\cdot \R \\
    &\geq \frac{1}{e}\cdot\Aux + \left(\frac{5}{4e}-\frac{1}{2e^2}\right)\cdot \V\\
    &=\V\cdot \left(\frac{1}{e}\cdot \chi_F+\frac{5}{4e}-\frac{1}{2e^2}\right),
\end{align*}
as desired. In the second equality, we have used the fact $\int_0^{\infty}{F(z)\cdot f(z)\ud{z}}=1/2$.
\end{proof}

By combining the two lemmas above and recalling that the expected revenue of mechanism $M_1$ lower-bounds that of mechanism $\mpi$, we obtain the following guarantee.

\begin{theorem}\label{thm:robustness-of-m2}
The robustness of mechanism $\mpi$ when applied on a strategic bidder and a profiled agent (with a possibly incorrect prediction) who draw their valuations from the monotone hazard rate probability distribution $F$ is at most $\frac{4e+2}{5}\approx 2.575$.
\end{theorem}

\begin{proof}
    Recall that, by Lemma~\ref{lem:optimal-rev}, the optimal expected revenue achieved by mechanism $\moc$ when applied on a strategic and an honest bidder who draw their valuations from distribution $F$ is $\V\cdot (\chi_F+1)$. Using the revenue guarantees from Lemmas~\ref{lem:m1-rev-1} and~\ref{lem:m1-rev-2}, we get a robustness of
    \begin{align*}
    \max\left\{\frac{\chi_F+1}{\chi_F+\frac{3}{4e}}, \frac{\chi_F+1}{\frac{1}{e}\cdot \chi_F+\frac{5}{4e}-\frac{1}{2e^2}}\right\}&=\frac{4e+2}{5}    
    \end{align*}
    for mechanism $M_1$, completing the proof. Notice that the two expressions inside the $\max$ are decreasing and increasing in $\chi_F$, respectively. Hence, the maximum among the two values is obtained when the two expressions become equal, i.e., for $\chi_F=\frac{1}{2e}$.
\end{proof}

\section{Robustness trade-offs}\label{sec:tradeoff}
Inspecting carefully the proof of Theorem~\ref{thm:robustness-of-m-star}, we get that the robustness of mechanism $\moc$ is $\frac{\chi_F+1}{\chi_F+\frac{1}{4e}}$. This quantity takes its worst-case value of $4e\approx 10.873$ for $\chi_F=0$ and is decreasing in $\chi_F$. In contrast, inspecting the proof of Theorem~\ref{thm:robustness-of-m2}, we find that the robustness of mechanism $\mpi$ is at most $\frac{\chi_F+1}{\frac{1}{e}\cdot \chi_F+\frac{5}{4e}-\frac{1}{2e^2}}$ for $\chi_F\in \left[0,\frac{1}{2e}\right]$. This expression is increasing in $\chi_F$ from the value of $\frac{4e^2}{5e-2}\approx 2.55$ to the worst-case value of $\frac{4e+2}{5}\approx 2.575$ for $\chi_F=\frac{1}{2e}$. 

In the following, we present better bounds on the Pareto frontier defined by the robustness bounds of the two mechanisms. For mechanism $\mpi$, we can prove better robustness bounds by bounding its expected revenue by that of the mechanism $M_2$, defined as follows. Mechanism $M_2$ selects a random sample $z$ from $F$ and offers the item to the strategic agent at a price of $z$. If the strategic agent does not buy the item, the mechanism offers the item to the profiled agent at the price $\phi^{-1}(0)$.

Denoting by $\Rtwo$ the expected revenue of mechanism $M_2$, we have
\begin{align}\nonumber
\Rtwo &= \int_0^{\infty}{\left(z\cdot (1-F(z))+F(z)\cdot \R\right)\cdot f(z)\ud{z}}=\Q+\R\cdot \int_0^{\infty}{F(z)\cdot f(z)\ud{z}}\\\label{eq:rtwo}
&=\Q+\R/2\geq \left(\frac{1}{4}+\frac{1}{2e}\right)\cdot \V.
\end{align}
The inequality follows by Lemma~\ref{lem:r-q-vs-v}. Using (\ref{eq:rtwo}), the robustness of mechanism $M_2$ becomes $\frac{\chi_F+1}{\frac{1}{4}+\frac{1}{2e}}$, which is also increasing in $\chi_F$, taking its smallest value of $\frac{4e}{e+2}\approx 2.304$ for $\chi_F=0$. The robustness bound is better than that obtained for mechanism $M_1$ for $\chi_F\in \left[0,\frac{e}{4}-\frac{3}{4}+\frac{1}{2e}\right]$. 

The Pareto frontier defined by the robustness bounds of the two mechanisms $\moc$ and $\mpi$ is depicted in Figure~\ref{fig:pareto-frontier}. The robustness values for mechanism $\mpi$ are the best among those proved for mechanisms $M_1$ and $M_2$.

\begin{figure}[h]
    \centering
    \includegraphics[width=0.6\textwidth]{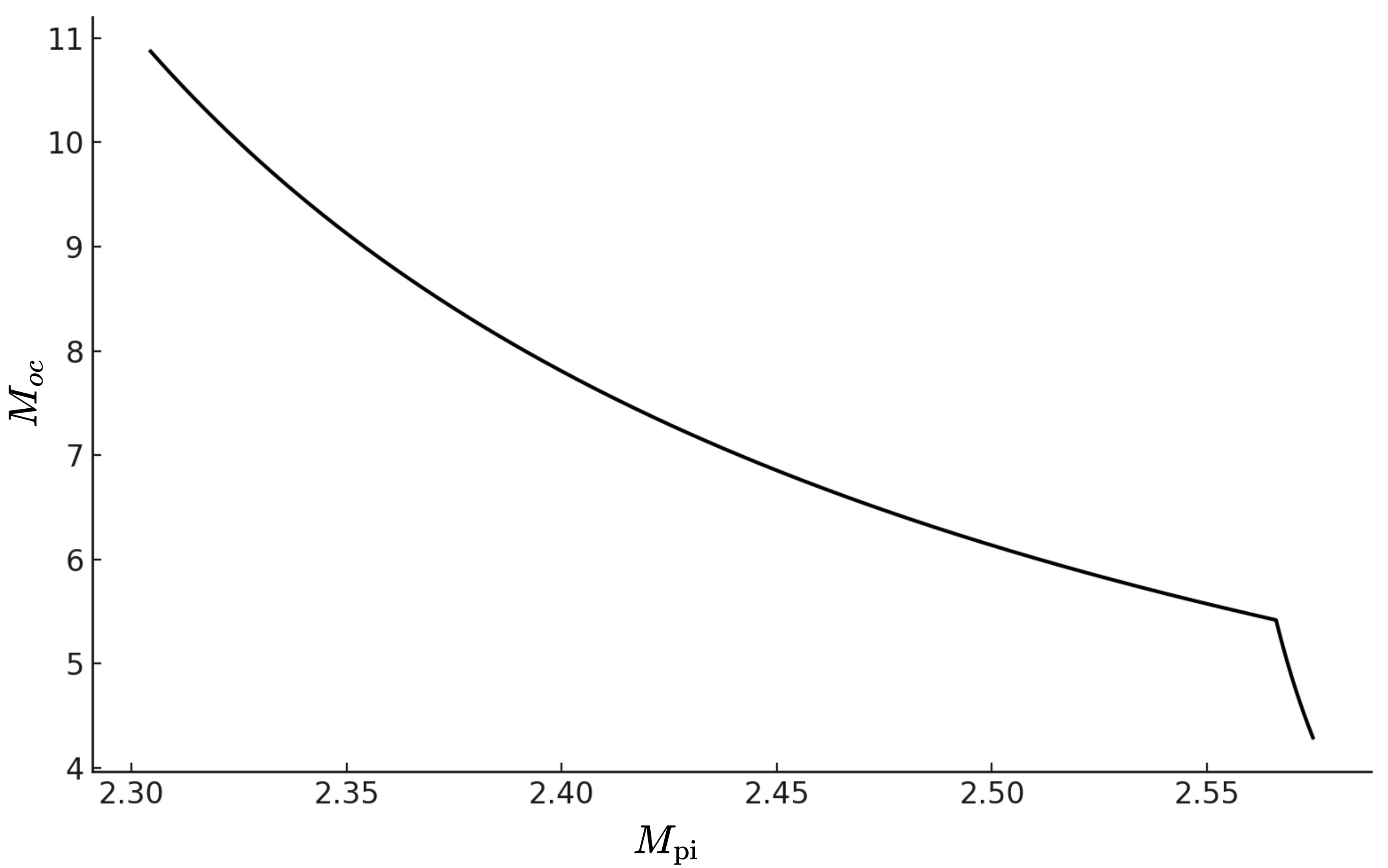}
    \caption{The Pareto frontier defined by the robustness bounds of the two mechanisms $\moc$ and $\mpi$ for $\chi_F\in \left[0,\frac{1}{2e}\right]$. For mechanism $\moc$, we use the robustness bound $\frac{\chi_F+1}{\chi_F+\frac{1}{4e}}$ from the proof of Theorem~\ref{thm:robustness-of-m-star}. For mechanism $\mpi$, we use the combination $\min\left\{\frac{\chi_F+1}{\frac{1}{e}\cdot \chi_F+\frac{5}{4e}-\frac{1}{2e^2}}, \frac{\chi_F+1}{\frac{1}{4}+\frac{1}{2e}}\right\}$ of the robustness bounds for mechanism $M_1$ (from the proof of Theorem~\ref{thm:robustness-of-m2}) and $M_2$ (by the discussion in this section).}
    \label{fig:pareto-frontier}
\end{figure}

\section{Discussion}\label{sec:open}
We have initiated the study of auction design with predictions in a Bayesian setting. We have kept our setting as simple as possible with a strategic bidder and a profiled agent, for whom a possibly incorrect prediction is available. 

Our robustness bounds for mechanisms $\moc$ and $\mpi$ are not known to be tight. For the exponential probability distribution with $F(z)=1-\exp(-\lambda\cdot z)$ with $\lambda>0$, Equations (\ref{eq:rev-strategic}) and (\ref{eq:rev-honest}) yield a robustness of $\frac{\Rs+\Rh}{\Rs}=\frac{4e+1}{3}\approx 4.291$ for mechanism $\moc$. Using these equations and calculating the expected revenue of mechanism $\mpi$ using (\ref{eq:rev-of-mpi}), we get a robustness of $\frac{e+1/2}{1+e^{-1/e}}\approx 1.902$. These are the only lower bounds we have for the worst-case robustness of $\moc$ and $\mpi$ among all MHR distributions. Closing the gaps between $10.873$ and $4.291$ and between $2.575$ and $1.902$ and determining the tight robustness bounds seems to require better analysis and tighter bounds on the inverse virtual valuation $\phi^{-1}(z)$.

One may wonder why we have restricted our attention to monotone hazard rate distributions and do not consider more general, e.g., regular, ones. The next statement shows that this restriction is necessary.

\begin{theorem}
The robustness of any mechanism can be unbounded when valuations are drawn from non-MHR regular probability distributions.
\end{theorem}

\begin{proof}
We will prove the theorem for the probability distribution that returns non-negative values with cumulative distribution function $F(z) = 1 - \frac{1}{(z + 1)^c}$ and density $f(z) = \frac{c}{(z+1)^{c+1}}$, for $c>1$. We can calculate the virtual valuation as
\begin{align*}
\phi(z) &= z - \frac{z + 1}{c} = z\cdot\bigg(1 - \frac{1}{c}\bigg) - \frac{1}{c}.
\end{align*}
Since $c>1$, the virtual valuation is increasing in $z$ and, thus, the distribution $F$ is regular. It is not MHR, though, since the hazard rate $\frac{f(z)}{1-F(z)}$ is equal to $\frac{c}{z+1}$, i.e., decreasing in $z$. Moving on, the inverse virtual valuation function is
\begin{align*}
\phi^{-1}(z) &= \frac{z + \frac{1}{c}}{1 - \frac{1}{c}} = \frac{z\cdot c + 1}{c - 1}.
\end{align*}
Hence,
\begin{align*}
F(\phi^{-1}(z)) &= 1 - \frac{1}{\left(\frac{z\cdot c + 1}{c - 1} + 1\right)^c}= 1 -\left(\frac{c-1}{c}\right)^c \cdot \frac{1}{(z+1)^c}
\end{align*}
Hence, we have that the maximum expected revenue we can get by a mechanism that cannot trust the prediction is at most
\begin{align*}
    2\cdot \R &= 2\cdot \phi^{-1}(0)\cdot \left(1-F(\phi^{-1}(0))\right)=\frac{2}{c}\cdot \left(\frac{c-1}{c}\right)^{c-1} \leq 2.
\end{align*}
We also have that the optimal expected revenue that can be extracted by a strategic bidder and an honest agent who both draw their valuations according to $F$ is
\begin{align*}
    \Rs+\Rh &=\Aux+\V \geq \V=\int_0^{\infty}{(1-F(z))\ud{z}}=\int_0^{\infty}{\frac{\ud{z}}{(z+1)^c}}=\frac{1}{c-1}.
\end{align*}
Hence, as $c$ approaches $1$, the optimal expected revenue goes to infinity and so does the robustness of any mechanism.
\end{proof}

Besides the obvious extension of the setting to more agents, we would like to understand whether better mechanisms (in terms of revenue consistency and robustness) can be designed when the prediction error is bounded by a known parameter.

\bibliographystyle{plainnat}
\bibliography{profiled}

\end{document}